\documentclass[12pt,a4paper]{article}
\usepackage[utf8]{inputenc}
\usepackage[hidelinks]{hyperref}
\usepackage{comment}
\usepackage{amsthm}
\usepackage[bottom]{footmisc}
\usepackage[papersize={8.3in,11.7in},vmargin={1in,1in},hmargin={1in,1in}]{geometry}
\usepackage{booktabs}
\usepackage{bm} % bold math
\usepackage{mathtools}
\usepackage{natbib}
\usepackage{amsfonts}
\usepackage{newtxtext,newtxmath}
\usepackage{setspace}

\theoremstyle{plain}

\newtheorem{theorem}{Theorem}

\newtheorem{lemma}[theorem]{Lemma}

\newtheorem{assumption}{Assumption}
\usepackage{bm} % For bold math symbols
\newcommand{\expectation}{\mathbb{E}}
\newcommand{\emexpectation}{\mathop{\mathbb{E}_n}}

\DeclareMathOperator*{\esssup}{ess\,sup}
\usepackage{physics}
\newcommand{\keywords}[1]{%
  \vspace{1em}%
  \noindent\textbf{Keywords:} #1
}
\newcommand{\eref}[1]{\null(\ref{#1})\null}
\usepackage{mathtools}
\DeclarePairedDelimiter{\ceil}{\lceil}{\rceil}

\DeclareMathOperator*{\argmin}{arg\,min} % thin space, limits underneath in displays
\theoremstyle{definition}
\newtheorem{definition}{Definition}
\newtheorem*{example}{Example}
\newtheorem{remark}{Remark}
\title{NEW APPROXIMATION RESULTS AND OPTIMAL ESTIMATION FOR FULLY CONNECTED DEEP NEURAL NETWORKS} 

\author{
  Zhaoji Tang \\
  Department of Economics, University College London 
\thanks{Email: \texttt{zhaoji.tang.23@ucl.ac.uk. I thank Dennis Kristensen and Andrei Zeleneev for their thoughtful suggestions and patient supervision. All errors are my own.}}}
\begin{document}

\maketitle

\begin{abstract}
\citet{farrell2021deep} establish non-asymptotic high-probability bounds for general deep feedforward neural network (with rectified linear unit activation function) estimators, with \citet[Theorem 1]{farrell2021deep} achieving a suboptimal convergence rate for fully connected feedforward networks. The authors suggest that improved approximation of fully connected networks could yield sharper versions of \citet[Theorem 1]{farrell2021deep} without altering the theoretical framework. By deriving approximation bounds specifically for a narrower fully connected deep neural network, this note demonstrates that \citet[Theorem 1]{farrell2021deep} can be improved to achieve an optimal rate (up to a logarithmic factor). Furthermore, this note briefly shows that deep neural network estimators can mitigate the curse of dimensionality for functions with compositional structure and functions defined on manifolds.
\end{abstract}

%%%%%%%%%%%%%%%%%%%%%%%%%%%%%%%%%%%%%%%%%%%%%%%%%%%%%%%%%%%%%%%%%%%%%%%%%
%%%% Main text entry area:
%%%%%%%%%%%%%%%%%%%%%%%%%%%%%%%%%%%%%%%%%%%%%%%%%%%%%%%%%%%%%%%%%%%%%%%%%

\section{Introduction}\label{s1}
Focusing on the rectified linear unit (ReLU) activation function, \citet[Theorem 2]{farrell2021deep} establish non-asymptotic high probability bounds for general deep feedforward neural network estimators. When restricting to fully connected neural networks (multilayer perceptrons), \citet[Theorem 1]{farrell2021deep} achieve a suboptimal convergence rate in the sense of \citet{stone1982optimal}.\footnote{While they achieve optimal convergence rates under specific conditions (see \citet[Corollary 1]{farrell2021deep}), this result applies only to a restricted neural network architecture.}

This work demonstrates that the convergence rate for fully connected networks is in fact optimal (up to a logarithmic factor). This is achieved by solely modifying the approximation results in the proof of \citet[Theorem 1]{farrell2021deep} while leaving the other components of the proof unchanged. The proof of \citet[Theorem 1]{farrell2021deep} unnecessarily increases the complexity of the neural network class by bounding the number of weights in terms of the squared width. By maintaining their core strategy while constructing a narrower approximation network with square root width, we resolve this issue. Using our approximation idea, other works that do not achieve the optimal convergence rate can also attain optimality (e.g., \citet{feng2023over}, \citet{brown2024statistical}, \citet{zhang2024causal}, \citet{colangelo2025double}, \citet{zhang2024classification}, \citet{chronopoulos2023deep}, \citet{jiao2025deep}).

Our proof adapts the approach of \citet[Theorem 1]{yarotsky2017error} and translates it from a wide neural network to a narrower (fully connected) architecture. This narrow architecture is deliberately constructed to resolve the squared width problem, representing a novel contribution to the literature. The key to this construction is a property of ReLU activation functions identified in \citet{schmidt2020nonparametric}---their ability to preserve information flow unchanged. Following \citet{liu2022optimal}, we convert the results to a fully connected network without substantially altering the depth or width.

Unlike the other works that derive convergence rates specifically for fully connected neural networks (e.g., \citet{liu2022optimal}, \citet{schmidt2020nonparametric}, \citet{shen2021deep}), the approach in \citet{farrell2021deep} starts from bounds for general neural network classes and then applies fully connected network approximation results to these general bounds. As noted in \citet{farrell2021deep}, this strategy offers several advantages: it establishes high-probability bounds, accommodates general loss functions, and permits unbounded weights.

\iffalse
Although \citet{farrell2021deep} assumes bounded random variables, one can show that in the linear model with mean squared error loss, boundedness of the dependent variable is not required. In this setting, this approach imposes the weakest assumptions on the error term compared to the existing literature.\footnote{Albeit at a slightly lower probability guarantee, we only require the error term to have a conditional zero mean and uniformly bounded variance. In turn, we allow for non-identically distributed data.}
\fi

Regarding convergence rates, our results align with most of the existing works that contain an extra logarithmic term (e.g., \citet{bauer2019deep} and \citet{schmidt2020nonparametric}), except for \citet{liu2022optimal}.\footnote{\citet{bauer2019deep} focused on smooth activation functions rather than ReLU, whereas the results in \citet{schmidt2020nonparametric} depend on network sparsity.} Their work demonstrates that ReLU feedforward networks can achieve the optimal rate without logarithmic factor terms. While the proof technique from \citet{farrell2021deep} inherently retains these logarithmic terms, this approach provides a trade-off: although yielding a slightly slower convergence rate, it avoids the additional assumptions about independent variables' density function that underlie the log-free results.

An important topic in the theoretical deep learning literature, which is not addressed in \citet{farrell2021deep}, is why deep neural networks often outperform shallow neural networks in practice. Since deep networks are typically employed in high-dimensional settings, a well-known explanation is that, under certain conditions, they can overcome the curse of dimensionality, whereas shallow networks cannot.

The literature in this area can be broadly divided into two strands. First, several works show that deep networks can circumvent the curse of dimensionality for functions with a compositional structure (e.g., \citet{poggio2017and}, \citet{schmidt2020nonparametric}). Second, other studies demonstrate that for functions defined on a low-dimensional manifold, the curse of dimensionality can also be mitigated (e.g., \citet{shaham2018provable}, \citet{kohler2022estimation}). For detailed discussions, we refer the reader to the cited works.

\citet[Chapter 8]{petersen2024mathematical} provides concise approximation results for functions with a compositional structure and for functions on manifolds, with proofs based on \citet[Theorem 1]{yarotsky2017error}. An advantage of \citet[Theorem 2]{farrell2021deep} is that these approximation results can be applied directly to derive convergence rates. However, achieving optimal rates for fully connected networks requires modifying the approximation using our results. For this reason, we briefly present our findings on mitigating the curse of dimensionality in this note.

The rest of the paper is organized as follows. Section \ref{sec:2} introduces the settings and results in \citet{farrell2021deep} and then presents our theoretical results. Section \ref{sec:curse} provides results for mitigating the curse of dimensionality. Section \ref{sec:conclusion} concludes, and the proofs are provided in the Appendix. 

We use the following norms. Let $\mathbf{X} \in \mathbb{R}^d$ be a random vector with sample realizations $\mathbf{x}_i$, and let $\mathbf{x}$ denote a generic realization. For a function $g(\mathbf{x})$,  define
$\norm{g}_\infty \coloneqq \sup_{\mathbf{x}} \abs{g(\mathbf{x})}$,
$\norm{g}_{L_2(X)} \coloneqq \big(\mathbb{E}[g(\mathbf{X})^2]\big)^{1/2}$, and
$\norm{g}_n \coloneqq \big(\emexpectation[g(\mathbf{x}_i)^2]\big)^{1/2}$,
where $\emexpectation[\cdot]$ denotes the sample mean.

\section{Deep neural networks}\label{sec:2}
This section proceeds in two steps. First, we recap the framework and the approach of \citet{farrell2021deep}. Second, we present our alternative method, clarifying both the distinctions between these approaches and why \citet[Theorem 1]{farrell2021deep} do not achieve the optimal convergence rate.
\subsection{Settings and results in \citet{farrell2021deep}}
For random covariates $\mathbf{X}\in \mathbb{R}^d$ and a scalar outcome $Y$, let $\mathbf{Z}=(Y,\mathbf{X}')'\in \mathbb{R}^{d+1}$, with a realization denoted as $\mathbf{z}=(y,\mathbf{x}')'$.

Following \citet{farrell2021deep}, for a loss function $\ell(f,\mathbf{z})$ and a function space $\mathcal{F}$, the function of interest $f_*$ is defined as:
	\begin{equation*}
	f_* \coloneqq \argmin_{f\in \mathcal{F}} \expectation[\ell(f,\textbf{Z})].	
\end{equation*}

 We wish to estimate $f_*$ using ReLU feedforward neural networks. To this end, define the ReLU activation function as:
\begin{equation*}
    \sigma_R(x)\coloneqq \max(x,0).
\end{equation*}

For covariates $\mathbf{X} \in \mathbb{R}^d$, a general feedforward neural network consists of $d$ input units, a number of hidden computation units, and one output unit. Each hidden unit computes a linear combination of its inputs and applies an activation function. The output unit is also a computation unit, but it does not apply an activation function. A key feature of feedforward neural networks is that the hidden units are organized into ordered groups, called layers. Units in a given layer receive input only from the preceding layers and send output only to the subsequent layers, but not vice versa.\footnote{For further details on the general case, we refer readers to \citet{farrell2021deep}.}

We denote the class of general feedforward neural networks by $\mathcal{F}_{\mathrm{DNN}}$, and the loss function satisfies the following assumption.
\begin{assumption}\label{a:loss}
For constants \(c_1, c_2, C_\ell \in \mathbb{R}_+\),
and for any \(f \in \mathcal{F} \cup \mathcal{F}_{\mathrm{DNN}}\) and 
\(g \in \mathcal{F} \cup \mathcal{F}_{\mathrm{DNN}}\),
assume the loss function satisfies
	\begin{equation*} 
		 \abs{\ell(f,\mathbf{z})-\ell(g,\mathbf{z})}\leq C_l \abs{f(\textbf{x})-g(\textbf{x})}.
         \end{equation*}
For \(f_* \in \mathcal{F}\) and for any \(h \in \mathcal{F}_{\mathrm{DNN}}\), assume the curvature condition
         \begin{equation*}
			c_1 \expectation[(h-f_*)^2]\leq \expectation[\ell(h,\textbf{Z})] - \expectation[\ell(f_*, \textbf{Z})]\leq c_2 \expectation[(h-f_*)^2].
	\end{equation*}
\end{assumption}

 The estimator of the general feedforward neural network is defined as:
	\begin{equation}
		\hat{f}_{DNN} \in \argmin_{\substack{f \in \mathcal{F}_{{\mathrm{DNN}}}\\  \norm{f}_\infty\leq 2M}} \sum_{i=1}^n \ell(f,\textbf{z}_i)\label{eq2}.
	\end{equation}

 We will focus on a special type of feedforward neural network, known as a fully connected neural network (MLP). A fully connected feedforward neural network can be represented as follows,
\begin{equation*}
    f_{\mathrm{MLP}}(\mathbf{x}) = \mathbf{W}_{L}  \boldsymbol{\sigma} \left(\cdots \boldsymbol{\sigma}\left(\mathbf{W}_{3} \boldsymbol{\sigma} \left(\mathbf{W}_{2}  \boldsymbol{\sigma} \left(\mathbf{W}_{1}  \boldsymbol{\sigma} \left(\mathbf{W}_{0}\mathbf{x} + \mathbf{b}_{0}\right) + \mathbf{b}_{1}\right) + \mathbf{b}_{2}\right) + \mathbf{b}_{3}\right)+\cdots \right) +\mathbf{b}_L,
\end{equation*}
where $L$ is a positive integer, $l=0,...,L$. $\mathbf{W}_l$ is defined as the $H_{l+1}\times H_l$ matrix, where $H_0=d$, $H_{L+1}=1$, $\mathbf{b}_l$ is the $H_l$-vector and $\boldsymbol{\sigma}:\mathbb{R}^{H_l}\rightarrow\mathbb{R}^{H_l}$ which applies $\sigma_R(\cdot)$ component-wise.

Denote this class by $\mathcal{F}_{\mathrm{MLP}}$.  Following the notation in \citet{farrell2021deep}, let 
\begin{equation}\label{eq:W}
  W  \coloneqq \sum_{l=0}^L \{H_l H_{l+1}+H_{l+1}\}
\end{equation}
denote the total number of weights (also called the size of the network), where $L$ represents the depth and $H$ the largest number of activation functions (referred to as computational nodes or units) $H_l$ for $l=1,...,l$. In practice, to reduce the number of tuning parameters, we choose $H_l$ to be the same for each $l$. We refer to $H$ as the width of the neural network.

Based on equation (\ref{eq:W}), we know that for each MLP, $\exists C\in \mathbb{R}_+$, such that
\begin{equation}
    W \leq C \cdot H^2L.
\end{equation}

\citet{farrell2021deep} make the following sampling assumption.

    \begin{assumption}\label{a:1}
        Assume that $\mathbf{z}_i=(y_i,\mathbf{x}_i')', 1\leq i\leq n$ are i.i.d. copies of $\mathbf{Z}=(Y,\mathbf{X})\in \mathcal{Y}\times[-1,1]^d$, where $\mathbf{X}$ is continuously distributed. For an absolute constant $M>0$, assume $\norm{f_*}_\infty \leq M$ and $\mathcal{Y}\subset [-M,M]$.
    \end{assumption}

Given the settings mentioned above, \citet{farrell2021deep} obtain the general results.
\begin{theorem}[General Feedforward Architecture]
\label{thm:general_architecture}
Suppose $f_*$ lies in a class $\mathcal{F}$. Suppose Assumption \ref{a:1} holds and define
\begin{equation*}
    \epsilon_{\mathrm{DNN}}\coloneqq \sup_{\Tilde{f}\in\mathcal{F}}\inf_{\substack{f \in \mathcal{F}_{{\mathrm{DNN}}}\\  \norm{f}_\infty\leq 2M}} \norm{f-\Tilde{f}}_\infty.
\end{equation*}
Let $\hat{f}_{\mathrm{DNN}}$ be the deep ReLU network estimator defined by equation (\ref{eq2}), for a loss function obeying Assumption \ref{a:loss}. Then with probability at least $1 - e^{-\gamma}$, for $n$ large enough,

\begin{enumerate}
    \item[(a)] $\|\hat{f}_{\mathrm{DNN}} - f_*\|_{L_2(X)}^2 \leq C\left(\frac{WL \log W}{n} \log n + \frac{\log \log n + \gamma}{n} + \epsilon_{\mathrm{DNN}}^2\right)$
    \item[(b)] $E_n[(\hat{f}_{\mathrm{DNN}} - f_*)^2] \leq C\left(\frac{WL \log W}{n} \log n + \frac{\log \log n + \gamma}{n} + \epsilon_{\mathrm{DNN}}^2\right)$
\end{enumerate}

\noindent for a constant $C > 0$ independent of $n$, which may depend on $d, M$, and other fixed constants.
\end{theorem}
One can treat $\frac{WL \log W}{n} \log n$ as the \emph{variance term} and $\epsilon_{\mathrm{DNN}}^2$ as the \emph{bias term}. This exhibits the classical \emph{bias-variance tradeoff}. Furthermore, since these are non-asymptotic bounds, there is an additional term $\frac{\gamma}{n}$ that corresponds to the \emph{confidence level}.

The function space for $f_*$ needed to be specified to derive $\epsilon_\mathrm{DNN}$. The assumption for the function space $\mathcal{F}$ of $f_*$ and the smoothness assumption is as follows.
\begin{assumption}\label{a:2}
		Assume $f_*$ lies in $\mathcal{F}$ which is the Sobolev ball $\mathcal{W}^{\beta,\infty}([-1,1]^d)$, with smoothness $\beta \in \mathbb{N}_+$,
		\begin{equation*}
			\mathcal{W}^{\beta,\infty}([-1,1]^d)\coloneqq \{f: \max_{{\boldsymbol\alpha},\abs{\boldsymbol\alpha}\leq \beta}{\esssup_{\mathbf{x}\in [-1,1]^d}} \abs{D^{\boldsymbol\alpha} f(\mathbf{x})}\leq 1\},
		\end{equation*}
		where $\boldsymbol{\alpha} = (\alpha_1,...,\alpha_d)$, $\abs{\boldsymbol\alpha}=\alpha_1+...+\alpha_d$ and $D^{\boldsymbol\alpha}f$ is the weak derivative.\footnote{Assumption \ref{a:2} requires $f_*\in[-1,1]$, which is not compatible with Assumption \ref{a:1}, where $f_*\in [-M,M]$. However, since $[-M,M]$ is bounded, the results extend directly. To remain consistent with \citet{farrell2021deep}, we follow their assumptions.}
	\end{assumption}

\citet[Theorem 1]{yarotsky2017error} provides an approximation result for a specific neural network architecture. \citet[Corollary 1]{farrell2021deep} applies this result to Theorem \ref{thm:general_architecture} and obtains a nearly optimal convergence rate. However, the result applies only to the particular architecture, and this architecture is difficult to construct in practice.

Directly converting an arbitrary architecture into a fully connected one yields a total number of weights of at least \( W \leq C \cdot H^2 L \). Unfortunately, this transformation generally increases \( W \), thereby breaking the variance-bias tradeoff. To address this issue, one needs to explicitly construct a narrower neural network.

Nevertheless, in \citet[Theorem 1]{farrell2021deep}, the authors convert the specific network constructed by \cite{yarotsky2017error} into a fully connected one using a conservative method (i.e. \citet[Lemma 1]{farrell2021deep}), which allows conversion to $MLP$ for any feedforward neural network. From this method, they derive the following results.
\begin{theorem}[Multilayer Perceptron]
		\label{th:suboptimal}
		Suppose Assumptions \ref{a:1} and \ref{a:2} hold. Let $\hat{f}_{\mathrm{MLP}}$ be the MLP ReLU network estimator defined by equation \eref{eq2}, restricted to $\mathcal{F}_{\mathrm{MLP}}$, for a loss function obeying Assumption \ref{a:loss}, with width $H\asymp n^{\frac{d}{2(\beta+d)}}\log^{2} n$ and depth $L\asymp \log n$. Then with probability at least $1-\exp(-n^{\frac{d}{\beta+d}}\log^{8} n )$, for n large enough,
		\begin{enumerate}
			\item[(a)] $\norm{\hat{f}_{\mathrm{MLP}}-f_*}^2_{L_2(\textbf{X})}\leq C\cdot \{n^{-\frac{\beta}{\beta+d}}\log^{8} n +\frac{\log\log n}{n}\}$ and
			\item[(b)] $\expectation_n[(\hat{f}_{\mathrm{MLP}}-f_*)^2]\leq C\cdot \{n^{-\frac{\beta}{\beta+d}}\log^{8} n +\frac{\log\log n}{n}\},$
		\end{enumerate}
		for a constant $C>0$ independent of $n$, which may depend on $d, \beta, M$, and other fixed constants.
	\end{theorem}
The choice of $H$ is larger than necessary such that the convergence rate for the MLP estimator is $\mathcal{O}(n^{-\frac{\beta}{\beta+d}}\log^8n)$, which is slower than the optimal convergence rate $\mathcal{O}(n^{-\frac{2\beta}{2\beta+d}})$ implied by \citet{stone1982optimal}.

    \subsection{Our approach}
The optimality of \citet[Corollary 1]{farrell2021deep} indicates that bounds in Theorem \ref{thm:general_architecture} and \citet[Theorem 1]{yarotsky2017error} are sharp. As mentioned above, the main issue arises from the steps taken to convert the approximation neural network in \citet[Theorem 1]{yarotsky2017error} into a fully connected architecture. 

On the one hand, the network in \citet[Theorem 1]{yarotsky2017error} is wide that converting it into a fully connected network inevitably and unnecessarily squares the total number of weights. On the other hand, the conversion method used in \citet{farrell2021deep} is conservative. They apply \citet[Lemma 1]{farrell2021deep} to convert the approximation network into a fully connected one. Their method unnecessarily increases $H_\mathrm{MLP}$ to be asymptotically equivalent to $W$ of the approximation network.

To address this issue, we modify \citet[Theorem 1]{yarotsky2017error} into the following theorem. In the first part, we reconstruct the network in \citet[Theorem 1]{yarotsky2017error} into a narrower architecture; in the second part, we apply Lemma \ref{cor:full} in the Appendix to convert it into a fully connected network.
\begin{theorem}\label{thm:ful}
    For any $\beta,d\in \mathbb{N_+}$ and $\epsilon_{\mathrm{DNN}}\in(0,1)$, 
    \begin{enumerate}
    \item there is a feedforward ReLU network architecture that
    it can approximate any function in $\mathcal{W}^{\beta,\infty}([-1,1]^d)$ with approximation error less than $\epsilon_{\mathrm{DNN}}$, with depth $L \leq c(ln(1/\epsilon_{\mathrm{DNN}})+1)$, and width $H\leq c\epsilon_{\mathrm{DNN}}^{-\frac{d}{2\beta}}$. 
    \item This approximation network can be converted into a fully connected network without significantly change $H,L$ and $W$. The only change is the constant $c$.  Using the fact that $W\leq C\cdot H^2L$ for fully connected neural networks, the total number of parameters satisfies $W_{\mathrm{MLP}}\leq c\epsilon_{\mathrm{DNN}}^{-\frac{d}{\beta}}(\ln(1/\epsilon_{\mathrm{DNN}})+1)$, with some constant $c=c(d,\beta)$.
    \end{enumerate}
\end{theorem}
\begin{remark}
    In \citet[Theorem 1]{yarotsky2017error}, the neural network was constructed via a linear combination of \( \mathcal{O}(\epsilon_{\mathrm{DNN}}^{-\frac{d}{\beta}}) \) subnetworks, requiring a width of \( \mathcal{O}(\epsilon_{\mathrm{DNN}}^{-\frac{d}{\beta}}) \). A direct fully connected implementation (where \( W \leq C \cdot H^2 L \)) would result in the number of weights at least \( \mathcal{O}(\epsilon_{\mathrm{DNN}}^{-\frac{2d}{\beta}}L)  \). This is the first reason why the rate established in Theorem \ref{th:suboptimal} is suboptimal. The first point solves this by explicitly constructing a narrower network with equivalent performance.
\end{remark}
\begin{remark}
    A narrower construction alone is not sufficient to solve the problem. Specifically, \citet[Lemma 1]{farrell2021deep} says $H_{\mathrm{MLP}}\asymp WL$, where $H_{\mathrm{MLP}}$ is the width of the converted MLP and $W,L$ are the total number of weights and depth of the original approximation network.\footnote{In fact, \citet[Lemma 1]{farrell2021deep} can be strengthened so that the width of the resulting fully connected network is bounded above by \( HL \), where \( H \) and \( L \) are the width and depth of the original architecture. However, this refinement still causes an undesirable expansion of \( H \) by an additional factor of \( L \).}  Since we maintain the same number of weights as in \citet[Theorem 1]{yarotsky2017error}, even with our narrower construction, applying \citet[Lemma 1]{farrell2021deep} would increase the width $H_{\mathrm{MLP}}\asymp  WL$, effectively squaring it. This is the second reason that Theorem \ref{th:suboptimal} is suboptimal.
\end{remark}
The proof is provided in Appendix \ref{proof:modified} and proceeds in two steps. The first step is the same as in the original proof from Yarotsky (2017), while the second step is a reconstruction of the approximating neural network. Although the overall strategy retains the partition of unity and Taylor expansion approach, we introduce significant structural refinements to reduce the network width while maintaining the same order of the coefficients. 

Compared with \citet[Theorem 1]{yarotsky2017error}, Theorem \ref{thm:ful} attains the same order of tue weights and depth. The only difference is that our construction achieves width $\mathcal{O}(\epsilon_{\mathrm{DNN}}^{-d/(2\beta)})$ instead of $\mathcal{O}(\epsilon_{\mathrm{DNN}}^{-d/\beta})$. In addition, we provide an explicit conversion to fully connected neural networks.

Choosing $\epsilon_{\mathrm{MLP}}=n^{-\frac{\beta}{2\beta+d}}, \gamma=n^{\frac{d}{2\beta+d}}\log^{4} n$, and combining Theorem \ref{thm:ful} with Theorem \ref{thm:general_architecture}, we obtain the following results for MLP estimators.
	\begin{theorem}[Optimal Multilayer Perceptron]
		\label{th:1}
		Suppose Assumtpions \ref{a:1} and \ref{a:2} hold. Let $\hat{f}_{\mathrm{MLP}}$ be the MLP ReLU network estimator defined by equation \eref{eq2}, for a loss function obeying Assumption \ref{a:loss}, with width $H\asymp n^{\frac{d}{4\beta+2d}}$ and depth $L\asymp \log n$. Then with probability at least $1-\exp(-n^{\frac{d}{2\beta+d}}\log^{4} n )$, for n large enough,\footnote{The logarithmic term could potentially be improved, but for simplicity we maintain the same log term as in \citet[Corollary 1]{farrell2021deep}.}
		\begin{enumerate}
			\item[(a)] $\norm{\hat{f}_{\mathrm{MLP}}-f_*}^2_{L_2(\textbf{X})}\leq C\cdot \{n^{-\frac{2\beta}{2\beta+d}}\log^{4} n +\frac{\log\log n}{n}\}$ and
			\item[(b)] $\expectation_n[(\hat{f}_{\mathrm{MLP}}-f_*)^2]\leq C\cdot \{n^{-\frac{2\beta}{2\beta+d}}\log^{4} n +\frac{\log\log n}{n}\},$
		\end{enumerate}
		for a constant $C>0$ independent of $n$, which may depend on $d, \beta, M$, and other fixed constants.
	\end{theorem}

Here, we obtain a convergence rate (up to logarithmic factors) of $n^{-\frac{2\beta}{2\beta+d}}$, rather than $n^{-\frac{\beta}{\beta+d}}$. This is the optimal rate in the sense of \citet{stone1982optimal}. 

This result contributes to the strand of literature building on the framework of \citet{farrell2021deep} (e.g., \citet{brown2024statistical}, \citet{zhang2024causal}, \citet{zhang2024classification}, \citet{colangelo2025double}). Moreover, by leveraging our approximation result, other studies that build on \citet{yarotsky2017error} can also attain the optimal convergence rate for statistical inference (e.g., \citet{feng2023over}, \citet{jiao2025deep}). Our approximation construction may also be useful for approximating other deep neural network architectures.

\section{Mitigating the curse of dimensionality}\label{sec:curse}
This section presents concise results illustrating scenarios in which deep neural networks can mitigate the curse of dimensionality. The results are obtained by directly applying the approximation results to Theorem \ref{thm:general_architecture}. The approximation results used here are adapted from \citet[Chapter 8]{petersen2024mathematical}.\footnote{Although \citet[Chapter 8]{petersen2024mathematical} derive approximation results on H\"older spaces, their arguments rely on a variant of \citet{yarotsky2017error}. Hence, the same reasoning directly extends to Sobolev spaces. Consequently, the results in this note also apply to H\"older spaces.} The key difference is that we convert the approximation networks into narrower and fully connected architectures using Theorem \ref{thm:ful}. Subsection \ref{sec:composite} addresses functions with compositional structures, while Subsection \ref{sec:manifold} focuses on functions defined on manifolds.
\subsection{Functions with compositional structure}\label{sec:composite}
We modify Assumption \ref{a:2} to the following compositional function assumption, based on \citet[Chapter 8.2]{petersen2024mathematical}.

Let $\mathcal{A}$ be a directed acyclic graph with $T$ vertices $\eta_1,\dots,\eta_T$ satisfying the following: the first $d$ vertices $\eta_1,\dots,\eta_d$ have no incoming edges; each vertex has at most $d_* \in \mathbb{N}$ incoming edges; and the final vertex $\eta_T$ has no outgoing edges. 

For every vertex $\eta_j$ with $j > d$, assign a function $f_j : \mathbb{R}^{d_j} \to \mathbb{R}$, where $d_j$ is the number of elements in the set
\begin{equation*}
S_j := \{\, i : \text{there is an edge from } \eta_i \text{ to } \eta_j \,\}.
\end{equation*}
Assume $1 \le d_j \le d_*$ for all $j > d$, and define
\begin{equation}\label{eq:Fj}
\begin{split}
F_j &:= x_j, \quad \text{for all } j \le d,\\
F_j &:= f_j\bigl( (F_i)_{i \in S_j} \bigr), \quad \text{for all } j > d. 
\end{split}
\end{equation}
Then $F_T(x_1,\dots,x_d)$ is a function from $\mathbb{R}^d$ to $\mathbb{R}$. Assuming
\begin{equation}\label{eq:comp}
\max_{{\boldsymbol\alpha},\abs{\boldsymbol\alpha}\leq \beta}{\esssup_{\mathbf{x}\in \mathbb{R}^{d_j}}} \abs{D^{\boldsymbol\alpha} f_j(\mathbf{x})}\le 1, \quad \text{for all } j = d+1,\dots,T, 
\end{equation}
we denote the set of all such functions $F_T$ by $\mathcal{W}^{\beta,\infty}_{d_*,T}([-1,1]^d)$. 
\begin{assumption}\label{a:composite}
Assume $f_*\in \mathcal{W}^{\beta,\infty}_{d_*,T}([-1,1]^d)$.
	\end{assumption}
This model of compositional functions can accommodate a broad class of compositional structures. To illustrate this idea, we present the interaction models from \citet{stone1994use}, noting that we allow for more general models than this simple summation case.

\begin{example}
    For some $d_* \in \{1,..,d\}$, denote $\mathbf{x}_I=(x_{i_1},...,x_{i_{d_*}})'$, where $I=\{i_1,...,i_{d_*}\}$ and $1\leq i_1\leq \cdots\leq i_{d_*}\leq d$. And assume the structure
    \begin{equation*}
        f_0(\mathbf{x})=\sum_{I\subseteq \{1,...,d\},\abs{I}=d_*} f_I(\mathbf{x}_I),
    \end{equation*}
   where each function $f_I(\cdot)$ satisfies equation~(\ref{eq:comp}). Denote the set of indices \( I \) as \(\mathcal{I}\), and let its cardinality be \( k \). Order the indices as \((I_1, \ldots, I_k)\). According to the definition in equation~(\ref{eq:Fj}), we have 
\begin{equation*}
F_{j+d} = f_{I_j} \quad \text{for } j = 1, \ldots, k.
\end{equation*}
Furthermore, define

\begin{equation*}
\begin{split}
    F_{d+k+1}& = f_{I_1} + f_{I_2},\\
    F_{d+k+j}& = F_{d+k+1} + f_{I_{j+1}}, \quad \text{for } j = 2, \ldots, k-1,\\
    F_{d+2k-1}& = f_0(\mathbf{x}).
\end{split}
\end{equation*}
Thus, we conclude that $f_0 \in \mathcal{W}^{\beta,\infty}_{d_*, d+2k-1}([-1,1]^d)$.
\end{example}

Applying the approximation result for functions with compositional structure (see Appendix \ref{appendix:curse}), we obtain the following theorem.
\begin{theorem}
		\label{th:composite}
		Suppose Assumptions \ref{a:1} and \ref{a:composite} hold. Let $\hat{f}_{\mathrm{MLP}}$ be the MLP ReLU network estimator defined by equation \eref{eq2}, restricted to $\mathcal{F}_{\mathrm{MLP}}$, for a loss function obeying Assumption \ref{a:loss}, with width $H\asymp n^{\frac{d_*}{2\beta+d_*}}\log^{2} n$ and depth $L\asymp \log n$. Then with probability at least $1-\exp(-n^{\frac{d_*}{2\beta+d_*}}\log^{4} n )$, for n large enough,
		\begin{enumerate}
			\item[(a)] $\norm{\hat{f}_{\mathrm{MLP}}-f_*}^2_{L_2(\textbf{X})}\leq C\cdot \{n^{-\frac{2\beta}{2\beta+d_*}}\log^{4} n +\frac{\log\log n}{n}\}$ and
			\item[(b)] $\expectation_n[(\hat{f}_{\mathrm{MLP}}-f_*)^2]\leq C\cdot \{n^{-\frac{2\beta}{2\beta+d_*}}\log^{4} n +\frac{\log\log n}{n}\},$
		\end{enumerate}
		for a constant $C>0$ independent of n, which may depend on $d_*,M,\beta,T$, and other fixed constants, but it doesn't depend on $d$.
	\end{theorem}
Here, the convergence rate depends on the largest dimension $d_*$ of the inputs within the compositional structure, rather than the original input dimension $d$. 

It is worth noting that the constant $C$ does not depend on $d$ (although $T$ may potentially depend on $d$), and this result may be particularly useful in high-dimensional settings. Under this setting, deep neural networks are superior to other machine learning methods in two respects. First, they can capture nonlinear compositional structures. Second, deep neural networks can automatically mitigate the curse of dimensionality without requiring variable selection or dimension reduction, thereby avoiding the bias that such procedures may introduce.

\subsection{Functions on manifolds}\label{sec:manifold}
The following definition of a compact smooth manifold is adapted from \citet[Chapter 1]{milnor1997topology}. For a comprehensive treatment of manifolds, we refer the reader to that work.
\begin{definition}[Compact Smooth Manifold]
        For $k,l\in \mathbb{N}_+$, let $U\subset \mathbb{R}^k$ and $V\subset \mathbb{R}^k$. A mapping $f:U\rightarrow V$ is called \textbf{smooth} if for each $\mathbf{x}\in U$ there exists an open set $U'\subset \mathbb{R}^k$ such that $\mathbf{x}\in U'$. And there exists a smooth mapping $F: U'\rightarrow \mathbb{R}^l$ such that $\forall \mathbf{x}^1\in U\cap U'$, $F(\mathbf{x}^1)=f(\mathbf{x}^1)$.
        
      A map $f:U\rightarrow V$ is called a \textbf{diffeomorphism} if $f$ maps $U$ homeomorphically onto $V$ and if both $f$ and its inverse $f^{-1}$ are smooth.
      
For $d,m\in \mathbb{N}_+$, a compact subset $\mathcal{M} \subset \mathbb{R}^d$ is called a \textbf{compact smooth manifold} of dimension $m$ if for every $\mathbf{x} \in \mathcal{M}$, there exists an open set $W\subset \mathbb{R}^k$ such that $\mathbf{x}\in W$ and $W\cap M$ is diffeomorphic to an open subset $U\subset\mathbb{R}^m$.
\end{definition}
We adjust Assumption \ref{a:2} to the following assumption.
\begin{assumption}\label{a:manifold}
 We further assume that there exists a compact, smooth $ m$-dimensional manifold $\mathcal{M}\subset [-1,1]^d$ such that $\mathbf{X}\in \mathcal{M}$. And
		we assume $f_*$ that lies in the Sobolev ball $\mathcal{W}^{\beta,\infty}(\mathcal{M})$, with smoothness $\beta \in \mathbb{N}_+$,
		\begin{equation*}
			\mathcal{W}^{\beta,\infty}(\mathcal{M})\coloneqq \{f: \max_{{\boldsymbol\alpha},\abs{\boldsymbol\alpha}\leq \beta}{\esssup_{\mathbf{x}\in \mathcal{M}}} \abs{D^{\boldsymbol\alpha} f(\mathbf{x})}\leq 1\},
		\end{equation*}
		where $\boldsymbol{\alpha} = (\alpha_1,...,\alpha_d)$, $\abs{\boldsymbol\alpha}=\alpha_1+...+\alpha_d$ and $D^{\boldsymbol\alpha}f$ is the weak derivative.
	\end{assumption}
Applying the approximation result for functions on manifolds (see Appendix \ref{appendix:curse}), we obtain the following theorem.
\begin{theorem}
		\label{th:manifold}
		Suppose Assumptions \ref{a:1} and \ref{a:manifold} hold. Let $\hat{f}_{\mathrm{MLP}}$ be the MLP ReLU network estimator defined by equation \eref{eq2}, restricted to $\mathcal{F}_{\mathrm{MLP}}$, for a loss function obeying Assumption \ref{a:loss}, with width $H\asymp n^{\frac{m}{2\beta+m}}\log^{2} n$ and depth $L\asymp \log n$. Then with probability at least $1-\exp(-n^{\frac{m}{2\beta+m}}\log^{4} n )$, for n large enough,
		\begin{enumerate}
			\item[(a)] $\norm{\hat{f}_{\mathrm{MLP}}-f_*}^2_{L_2(\textbf{X})}\leq C\cdot \{n^{-\frac{2\beta}{2\beta+m}}\log^{4} n +\frac{\log\log n}{n}\}$ and
			\item[(b)] $\expectation_n[(\hat{f}_{\mathrm{MLP}}-f_*)^2]\leq C\cdot \{n^{-\frac{2\beta}{2\beta+m}}\log^{4} n +\frac{\log\log n}{n}\},$
		\end{enumerate}
		for a constant $C>0$ independent of n, which may depend on $\mathcal{M}, \beta,m$, and other fixed constants.
\end{theorem}
In this case, the convergence rate depends on the intrinsic dimension $m$ rather than the ambient dimension $d$. Notice that the constant term $C$ implicitly depends on $d$ since $\mathcal{M}$ may implicitly depend on $d$.

In empirical economics applications, it is often difficult to justify that the function of interest possesses a compositional structure. In contrast, it is usually more feasible to assess whether the data lie on a low-dimensional manifold. Consequently, this result may be more useful in practice.

\section{Conclusion}\label{sec:conclusion}
By modifying only the approximation component of the proof in \citet[Theorem 1]{farrell2021deep}, this note establishes that the convergence rate of MLP estimators is indeed optimal (up to a logarithmic factor) in the sense of \citet{stone1982optimal}. To this end, this note extends the results of \citet{yarotsky2017error}. Instead of focusing on approximation results tied to network size, we demonstrate that a narrower network of comparable size can achieve the same approximation rate. This result can be of independent interest. One can apply our narrower network construction to convert other neural architectures into fully connected networks without substantially increasing the total number of parameters.

In addition, we establish convergence rates for functions with compositional structures and for functions defined on low-dimensional manifolds. In both cases, consistent with previous literature, we show that deep neural networks can effectively circumvent the curse of dimensionality. We encourage future empirical work to justify the use of deep learning based on these two reasons.

%%%%%%%%%%%%%%%%%%%%%%%%%%%%%%%%%%%%%%%%%%%%%%
%% Example with multiple Appendixes:        %%
%%%%%%%%%%%%%%%%%%%%%%%%%%%%%%%%%%%%%%%%%%%%%%
\begin{appendix}

\section{Lemma}

\subsection{Lemmas used in Appendix \ref{proof:modified}}
We will first introduce the lemma from \citet{yarotsky2017error}, which combines results from \cite{liu2022optimal}. We then introduce a lemma used to convert approximation neural networks into fully connected networks. These results are used to show the main results in Appendix \ref{proof:modified}.

The following lemma is based on \citet[Proposition 3]{yarotsky2017error}, while its final part follows from \citet[Proposition 2]{liu2022optimal}.
\begin{lemma}\label{le:neural}
    The function $f(x)=x^2$ on the segment $[0,1]$ can be approximated with any error $\delta>0$ by a ReLU network $\Tilde{f}_{sq,\delta}$ such that $\Tilde{f}_{\mathrm{sq},\delta}(0)=0$ and $\abs{\Tilde{f}_{\mathrm{sq},\delta}-x^2}< \delta$ for $x\in [-1,1]$. Given $Q>0$, and $\epsilon \in (0,1)$, there is a ReLU network $\eta$ with two input units that implements a function $\Tilde{\times} : \mathbb{R}^2 \rightarrow \mathbb{R}$.
    
Assume without loss of generality that $Q\geq 1$ and set
    \begin{equation*}\label{eq:times}
        \Tilde{\times}(x,y)=\frac{Q^2}{8} (\Tilde{f}_{\mathrm{sq},\delta}(\frac{\abs{x+y}}{2Q})-\Tilde{f}_{\mathrm{sq},\delta}(\frac{\abs{x}}{2Q})-\Tilde{f}_{\mathrm{sq},\delta}(\frac{\abs{y}}{2Q})),
    \end{equation*}
    where $\delta=\frac{8\epsilon}{3Q^2}$. Then, the following hold:

    \begin{enumerate}
        \item for any inputs $x,y$, if $\abs{x}\leq  Q$ and $\abs{y}\leq Q$, then $\abs{\Tilde{\times}(x,y)-xy}\leq \epsilon$;
        \item $\Tilde{\times}(x,0)=\Tilde{\times}(0,y)=0$;
        \item the depth and the number of weights and hidden units in $\eta$ are not greater than $c_1 \ln(1/\epsilon)+c_2$ with an universal constant $c_1$ and a constant $c_2 = c_2(Q)$,
        \item the neural network can be made fully connected with a width of at most 12.
    \end{enumerate}
    
\end{lemma}

The next lemma enables us to convert the approximating function used in Proof \ref{proof:modified} into fully connected architectures.
\begin{lemma}\label{cor:full}
    With $\Tilde{\times}(\cdot)$ defined in Lemma \ref{le:neural}, and a fully connected feedforward ReLU neural network is denoted by $f_\mathit{MLP}^m$ or $f_\mathit{MLP}^\alpha$, indicating dependence on the parameter $m$ or $\alpha$, respectively. We assume that, for parameters $m,\alpha\in \mathbb{R}$, $m$ takes values in a finite set $S_m$ and $\alpha$ takes values in a finite set $S_\alpha$. Denote $N_{\mathrm{max}}\coloneqq \max(\abs{S_m},\abs{S_\alpha})$. Suppose each of the neural networks has a depth at most $L$, a width at most $H$, and $a_{m,\alpha} \in \mathbb{R}$. Then, we can convert the following neural networks $\sum_{\alpha\in S_\alpha}\Tilde{\times}(\sum_{m\in S_m} a_{m,\alpha}f_\mathit{MLP}^m,f_\mathit{MLP}^\alpha)$ into fully connected architectures with depth at most $L+1+c_1 \ln(1/\epsilon)+c_2$ and width at most $\max(12,2H) N_{\mathrm{max}}$.\footnote{For simplicity, we assume the functions satisfies the definition of $\Tilde{\times}(\cdot)$.}
    \begin{remark}
    A key observation for our proof is that the input can always be passed through a ReLU activation function without altering its value, which is not possible with other activation functions (e.g., the sigmoid function).
\end{remark} 
    \begin{proof}
        For each $f_\mathrm{MLP}^m, m\in S_m$, if the depth is less than $L$, we can use the fact that
        \begin{equation*}
            x = \sigma_R(x)-\sigma_R(-x)
        \end{equation*}
       to extend it so that every $f_\mathrm{MLP}^m$ has the same depth $L$. Since each $f_\mathrm{MLP}^m$ is fully connected and has the same depth, combining them into a single fully connected network is achieved by simply connecting each node to the nodes in the adjacent layers. Thus, each value of $f_\mathrm{MLP}^m$ for $m\in S_m$ can be represented by the fully connected neural network with depth at most $L$ and width at most $2HN_\mathrm{max}$ with $\abs{S_m}$ output units. The same arguments apply to $f_\mathrm{MLP}^\alpha$.

        For each $\alpha$, the only difference of $\sum_{m\in S_m} a_{m,\alpha}f_\mathrm{MLP}^m$ is $a_{m,\alpha}$. They are merely linear combinations of $f_\mathrm{MLP}^m$ for $m\in S_m$. Therefore, with one more layer, we can use the same fully connected neural network but with at most $N_\mathrm{max}$ output units to represent each value of $\sum_{m\in S_m} a_{m,\alpha}f_\mathrm{MLP}^m$. Each output unit corresponds to a different value of $\alpha$. That is,

        \iffalse
         And
        \begin{equation*}
            \sum_{x\in S_x} a_{x,y}f_\mathrm{MLP}(x) =  \sum_{x\in S_x}a_{x,y}(\sigma_R(f_\mathrm{MLP}(x))-\sigma_R(-f_\mathrm{MLP}(x))).
        \end{equation*}

        Moreover, moving from a previous layer to the next layer, we can always keep $\sum_{m\in S_m} a_{m,\alpha}f_\mathrm{MLP}^m$ by the fact that\fi
        \begin{equation*}
            \sum_{m\in S_m} a_{m,\alpha}f_\mathrm{MLP}^m = \sigma_R(\sum_{m\in S_m} a_{m,\alpha}f_\mathrm{MLP}^m)-\sigma_R(-\sum_{m\in S_m} a_{m,\alpha}f_\mathrm{MLP}^m)
        \end{equation*}
        and the width needed for this operation is at most $2N_\mathrm{max}$.
     
        Thus, with a single fully connected neural network, we can represent each value of $\sum_{m\in S_m} a_{m,\alpha}f_\mathrm{MLP}^m$ and $f_\mathrm{MLP}^\alpha$ for each $\alpha\in S_\alpha$. The neural network has at most $2N_\mathrm{max}$ output units, depth at most $L+1$, and width at most $2HN_\mathrm{max}$. Each output unit corresponds either $\sum_{m\in S_m} a_{m,\alpha}f_\mathrm{MLP}^m$ or $f_\mathrm{MLP}^\alpha$ for each $\alpha\in S_\alpha$.

        For each $\alpha\in S_\alpha$, each $\sum_{m\in S_m} a_{m,\alpha}f_\mathrm{MLP}^m$ corresponds to a single $f_\mathrm{MLP}^\alpha$. Thus, by Lemma \ref{le:neural}, we need at most $L+1+c_1 \ln(1/\epsilon)+c_2$ more layers, whose width are at most $\max(12,2H)N_\mathrm{max}$, to represent $\sum_{\alpha\in S_\alpha}\Tilde{\times}(\sum_{m\in S_ m} a_{m,
        \alpha}f_\mathrm{MLP}^m,f_\mathrm{MLP}^\alpha)$.
        \end{proof}
        
\end{lemma}
\subsection{Lemmas used in Section \ref{sec:curse}}\label{appendix:curse}
The lemmas in this subsection are adapted from \citet[Chapter 8]{petersen2024mathematical}, with the condition on the number of weights replaced by a condition on width. Based on the proofs in \citet[Chapter 8]{petersen2024mathematical}, one can directly establish the following two lemmas. The only change is that, in determining the width, we apply our Theorem~\ref{thm:ful} rather than \citet[Theorem 1]{yarotsky2017error}, which is used there to determine the network size. For brevity, we state the two lemmas without proof.

The following lemma is modified from \citet[Proposition 8.5]{petersen2024mathematical} and used to show Theorem \ref{th:composite}.
\begin{lemma}\label{lemma:composite}
    For any $\beta, d,d_*,T \in \mathbb{N_+}$ and $\epsilon\in (0,1)$, there is a fully connected ReLU network architecture that can represent any function in $\mathcal{W}^{\beta,\infty}_{d_*,T}([-1,1]^d)$, with the depth $L \leq c(ln(1/\epsilon)+1)$, the width $H\leq c\epsilon^{-\frac{d_*}{2\beta}}$. Moreover, using the fact that $W\leq C\cdot H^2L$, the weights $W\leq c\epsilon^{-\frac{d_*}{\beta}}(\ln(1/\epsilon)+1)$, with some constant $c=c(d_*,T,\beta)$.
\end{lemma}

The next lemma is modified from \citet[Proposition 8.7]{petersen2024mathematical}, and used to show Theorem \ref{th:manifold}.
\begin{lemma}\label{lemma:manifold}
     For any $\beta, d,m \in \mathbb{N_+}$ and $\epsilon\in (0,1)$, suppose $\mathcal{M}$ is a smooth, compact $m$ dimensional manifold s.t. $\mathcal{M}\subseteq [-1,1]^d$. Then, there is a fully connected ReLU network architecture that can represent any function in $\mathcal{W}^{\beta,\infty}(\mathcal{M})$, with the depth $L \leq c(ln(1/\epsilon)+1)$, the width $H\leq c\epsilon^{-\frac{m}{2\beta}}$. Moreover, using the fact that $W\leq C\cdot H^2L$, the weights $W\leq c\epsilon^{-\frac{m}{\beta}}(\ln(1/\epsilon)+1)$, with some constant $c=c(d,m,\mathcal{M},\beta)$.
\end{lemma}

\section{Proof}\label{proof:modified}

For completeness, we provide a self-contained presentation below, with our modifications clearly highlighted. Appendix \ref{proof:error} reproduces \citet[Theorem 1]{yarotsky2017error}, while Appendix \ref{proof:modification} presents our modified proof based on Appendix \ref{proof:error}. We also include the proof of \citet[Theorem 1]{yarotsky2017error} so that readers can clearly see the connections and differences with our results.

\subsection{\citet[Theorem 1]{yarotsky2017error}}\label{proof:error}
The following lemma is \citet[Theorem 1]{yarotsky2017error}, while Theorem \ref{thm:ful} is a modified version of it.
\begin{lemma}
    For any $\beta,d\in \mathbb{N_+}$ and $\epsilon\in(0,1)$, there is a ReLU network architecture that
    it can approximate any function in $\mathcal{W}^{\beta,\infty}([-1,1]^d)$ with error $\epsilon$, with the depth $L \leq c(ln(1/\epsilon)+1)$, the number of weights $W\leq c\epsilon^{-\frac{d}{\beta}}ln(1/\epsilon+1)$, with some constant $c=c(d,\beta)$.
\end{lemma}
\begin{proof}
\begin{quotation}
    
    Given a $N\in \mathbb{N}_+$. Consider a partition of unity:
    \begin{equation*}
        \sum_\mathbf{m} \phi_\mathbf{m}(\mathbf{x})\equiv 1, \quad \mathbf{x}\in [-1,1]^d.
    \end{equation*}
    which is constructed by a grid of $(2N+1)^d$ function $\phi_\mathbf{m}$ on the domain $[-1,1]^d$

    Here $\mathbf{m}=(m_1,...,m_d), m_i\in \{-N,...-1,0,1,...,N\}, i=1,...,d$, and define the function the function $\phi_\mathbf{m}$ as 
    \begin{equation*}
        \phi_\mathbf{m}(\mathbf{x}) = \prod_{k=1}^d \psi(3N(x_k -\frac{m_k}{N})),
    \end{equation*}
    where
    \begin{equation*}
        \psi(x)=\begin{cases}
            1, & \abs{x}<1,\\
            0, & 2<\abs{x},\\
            2-\abs{x}, & 1\leq \abs{x}\leq 2.
        \end{cases}
    \end{equation*}
    
Note that
\begin{equation*}
    \norm{\psi}_\infty=1 \quad and  \quad \norm{\phi_\mathbf{m}}_\infty=1, \forall \mathbf{m}
\end{equation*}
and
\begin{equation*}
    \operatorname{supp} \phi_\mathbf{m} \subset \{\mathbf{x}:\abs{x_k-\frac{m_k}{N}}<\frac{1}{N}, \forall k\}.
\end{equation*}

For any $\mathbf{m}\in \{-N,...,-1,0,1,...,N\}^d$, consider the degree-($\beta-1$) Taylor polynomial for the function $f$ at $\mathbf{x}=\frac{\mathbf{m}}{N}$:
\begin{equation}\label{eq:P}
    P_\mathbf{m}(\mathbf{x}) = \sum_{{\boldsymbol{\alpha}}:\abs{\boldsymbol{\alpha}}<\beta}\frac{D^{{\boldsymbol{\alpha}}}f}{{\boldsymbol{\alpha}}!}\bigg|_{\mathbf{x} = \frac{\mathbf{m}}{N}}
\left( \mathbf{x} - \frac{\mathbf{m}}{N} \right)^{{\boldsymbol{\alpha}}},
\end{equation}
with ${\boldsymbol{\alpha}}!=\prod_{k=1}^d \alpha_k!$ and $(\mathbf{x}-\frac{\mathbf{m}}{N})^{\boldsymbol{\alpha}} = \prod_{k=1}^d (x_k-\frac{m_k}{N})^{\alpha_k}$. Define an approximation to $f$ by
\begin{equation}\label{eq:f1}
    f_1 = \sum_{\mathbf{m}\in \{-\mathbf{N},...,-1,0,1,...,\mathbf{N}\}} \phi_\mathbf{m} P_\mathbf{m}.
\end{equation}

Using the Taylor expansion of $f$, the approximation error is bounded by:
\begin{equation*}
    \begin{split}
        \abs{f(\mathbf{x})-f_1(\mathbf{x})} 
        &\leq \frac{2^d d^\beta}{\beta!}(\frac{1}{N})^\beta.
    \end{split}
\end{equation*}

$\forall x\in \mathbb{R}$, a ceiling function maps $x$ to the least integer greater than or equal to $x$, denoted $\ceil{x}$. Choose
\begin{equation}\label{eq:N}
    N= \ceil{(\frac{\beta!}{2^d d^\beta}\frac{\epsilon}{2})^{-1/\beta}},
\end{equation}
then
\begin{equation*}
    \norm{f-f_1}_\infty \leq \frac{\epsilon}{2}.
\end{equation*}

By equation (\ref{eq:P}) the coefficients of the polynomials $P_\mathbf{m}$ are uniformly bounded $\forall f \in \mathcal{W}^{\beta,\infty}([-1,1]^d)$:
\begin{equation}\label{eq:Pm}
    P_\mathbf{m}(\mathbf{x})=\sum_{{\boldsymbol{\alpha}}: \abs{\mathfrak{\alpha}}<\beta} a_{\mathbf{m},\boldsymbol{\alpha}}(\mathbf{x}-\frac{\mathbf{m}}{N})^{\boldsymbol{\alpha}}, \quad \abs{a_{\mathbf{m},\boldsymbol{\alpha}}}\leq 1.
\end{equation}

The second step is to construct a network architecture which can approximate any function of the form (\ref{eq:f1}) with uniform error $\frac{\epsilon}{2}$, assuming that $N$ is given by (\ref{eq:N}) and the polynomials $P_\mathbf{m}$ are of the form (\ref{eq:Pm}).

Expand $f_1$ as
\begin{equation*}
        f_1(\mathbf{x}) = \sum_{\mathbf{m}\in \{-N,...,-1,0,1,...,N\}^d} \sum_{{\boldsymbol{\alpha}}:\abs{\alpha}<\beta} a_{\mathbf{m},\boldsymbol{\alpha}} \phi_\mathbf{m}(\mathbf{x})(\mathbf{x}-\frac{\mathbf{m}}{N})^{\boldsymbol{\alpha}}.
\end{equation*}

The expansion is a linear combination of $\phi_\mathbf{m}(\mathbf{x})(\mathbf{x}-\frac{\mathbf{m}}{N})^\mathbf{n}$, and the total number of such terms is less than $d^\beta \prod_{i=1}^d(2N+1)$. Each of them is a product of at most $d+\beta-1$ piece-wise linear univariate factors: total number of $d$ for $\psi(3Nx_k-3m_k)$ functions and at most $\beta-1$ linear expressions of $x_k-\frac{m_k}{N}$. A neural network can approximate such a product according to Lemma \ref{le:neural}. 

Choose $Q=d+\beta$ and $\delta$ (to be chosen later), and use $\Tilde{\times}(\cdot)$ (defined in Lemma \ref{le:neural}) to denote the approximate multiplication and iteratively apply $\Tilde{\times}$ to approximate the product $\phi_\mathbf{m}(\mathbf{x})(\mathbf{x}-\frac{\boldsymbol{m}}{N})^\mathbf{n}$. That is,
\begin{equation}\label{eq:times}
    \Tilde{f}_{\mathbf{m},\boldsymbol{\alpha}}(\mathbf{x}) =\Tilde{\times}(\psi(3Nx_1-3m_1),\Tilde{\times}(\psi(3Nx_2-3m_2),...,\Tilde{\times}(\psi(3Nx_k-3m_k,...)...)).
\end{equation}
Since,
\begin{equation}\label{eq:approximation_error}
    \abs{\Tilde{f}_{\mathbf{m},\boldsymbol{\alpha}}(\mathbf{x})-\phi_\mathbf{m}(\mathbf{x})(\mathbf{x}-\frac{\mathbf{m}}{N})^{\boldsymbol{\alpha}}}\leq (d+\beta-1)\delta
\end{equation}
The full approximation is:
\begin{equation*}
    \Tilde{f}=\sum_{\boldsymbol{m}\in \{-N,...,-1,0,1,...,N\}^d}\sum_{\boldsymbol{\alpha}:\abs{\boldsymbol{\alpha}}<\beta}a_{\mathbf{m},\boldsymbol{\alpha}}\Tilde{f}_{\mathbf{m},\boldsymbol{\alpha}}.
\end{equation*}

There are at most $d^\beta (2N+1)^d$ possible values for $(\mathbf{m},{\boldsymbol{\alpha}})$.

The approximation error of $\Tilde{f}$ is:
\begin{equation*}
    \begin{split}
        \left| \Tilde{f}(\mathbf{x}) - f_1(\mathbf{x}) \right| 
        &\leq  2^d d^\beta (d + \beta-1) \delta.
    \end{split}
\end{equation*}

Choose
\begin{equation*}
    \delta = \frac{\epsilon}{2^{d+1}d^\beta (d+\beta-1)},
\end{equation*}
then $\norm{\Tilde{f}-f_1}_\infty\leq \frac{\epsilon}{2}$. And $\norm{\Tilde{f}-f}_\infty\leq \epsilon$.

\end{quotation}
\end{proof}
Since a full approximation neural network is a linear combination of fewer than $d^\beta \prod_{i=1}^d (2N+1)$ sub-networks, the resulting network width is $H = \mathcal{O}((2N+1)^d) = \mathcal{O}(\epsilon^{-\frac{d}{\beta}})$.

\subsection{Proof of Theorem \ref{thm:ful}}\label{proof:modification}
\begin{proof}
 Based on the proof above, first notice that we can directly represent $\psi(\cdot)$ by a ReLU neural network,
\begin{equation*}
    \psi(x)= \sigma_R(2-\sigma_R(x)-\sigma_R(-x))-\sigma_R(1-\sigma_R(x)-\sigma_R(-x)).
\end{equation*}

Since $\phi_{\mathbf{m}}(\cdot)$ is a product of $\psi(\cdot)$, we can easily approximate it using Lemma \ref{le:neural}. The fundamental constraint is that $a_{\mathbf{m},\boldsymbol{\alpha}}$ admits at most $d^\beta (2N+1)^d$ distinct values. Under this condition, a narrower network architecture can maintain identical approximation error bounds. This is achieved solely through an appropriate reconstruction of $\Tilde{f}_{\mathbf{m},\boldsymbol{\alpha}}(\mathbf{x})$.

Constrained by the bound $W\leq C\cdot H^2L$ for fully connected neural networks, our goal is to construct a neural network with width  $\mathcal{O}(2N+1)^{\frac{d}{2}}$. If $d$ is an even number, the modification is straightforward. For odd $d$, however, a more delicate construction is required. 

The proof is separated into three cases:

\textbf{Case 1: if $d$ is even}, we can write:
\begin{equation*}
\begin{split}
     &\phi_{\mathbf{m}}(\mathbf{x})(\mathbf{x}-\frac{\mathbf{m}}{N})^{{\boldsymbol{\alpha}}}\\
     =& \prod_{k=1}^d \psi(3N(x_k-\frac{m_k}{N}))(x_k-\frac{m_k}{N})^{\alpha_k}\\
     =&\left( \prod_{k=1}^{\frac{d}{2}} \psi(3N(x_k-\frac{m_k}{N}))(x_k-\frac{m_k}{N})^{\alpha_k}\right) \left( \prod_{k=\frac{d}{2}+1}^{d} \psi(3N(x_k-\frac{m_k}{N}))(x_k-\frac{m_k}{N})^{\alpha_k}\right).
\end{split}
\end{equation*}

Denote $\mathbf{m}^1=(m_1,...,m_{\frac{d}{2}}),\mathbf{m}^2 = (m_{\frac{d}{2}+1},...,m_d), {\boldsymbol{\alpha}}^1=(\alpha_1,...,\alpha_{\frac{d}{2}}), {\boldsymbol{\alpha}}^2=(\alpha_{\frac{d}{2}+1},...,\alpha_d)$, $\mathbf{x}^1=(x_1,...,x_{\frac{d}{2}})$, $\mathbf{x}^2=(x_\frac{d}{2},...,x_d)$.

We can rewrite $f_1(\mathbf{x})$, so that
\begin{align*}
        f_1(\mathbf{x}) &= \sum_{\mathbf{m}\in \{-N,...,-1,0,1,...,N\}^d} \sum_{{\boldsymbol{\alpha}}:\abs{\alpha}<\beta} a_{\mathbf{m},\boldsymbol{\alpha}} \phi_\mathbf{m}(\mathbf{x})(\mathbf{x}-\frac{\mathbf{m}}{N})^{\boldsymbol{\alpha}}\\
         &= \sum_{\mathbf{m}^2\in \{-N,...,-1,0,1,...,N\}^{\frac{d}{2}}} \sum_{{\boldsymbol{\alpha}^2}:\abs{\boldsymbol{\alpha}^2}<\beta} \phi_{\mathbf{m}^2}(\mathbf{x}^2)(\mathbf{x}^2-\frac{\mathbf{m}^2}{N})^{\boldsymbol{\alpha}^2}\\
         &\times\left(\sum_{\mathbf{m}^1\in \{-N,...,-1,0,1,...,N\}^{\frac{d}{2}}} \sum_{{\boldsymbol{\alpha}^1}:\abs{\boldsymbol{\alpha}^1}<\beta-\abs{\boldsymbol{\alpha}^2}} a_{\mathbf{m},\boldsymbol{\alpha}} \phi_{\mathbf{m}^1}(\mathbf{x}^1)(\mathbf{x}^1-\frac{\mathbf{m}^1}{N})^{\boldsymbol{\alpha}^1}\right).
\end{align*}

    Using the definition in equation (\ref{eq:times}), we can approximate $ \phi_{\mathbf{m}^1}(\mathbf{x}^1)(\mathbf{x}^1-\frac{\mathbf{m}^1}{N})^{\boldsymbol{\alpha}^1}$ by $\Tilde{f}_{\mathbf{m}^1,{\boldsymbol{\alpha}}^1}$ and $ \phi_{\mathbf{m}^2}(\mathbf{x}^2)(\mathbf{x}^2-\frac{\mathbf{m}^2}{N})^{\boldsymbol{\alpha}^2}$ by $\Tilde{f}_{\mathbf{m}^2,{\boldsymbol{\alpha}}^2}$. From Lemma \ref{le:neural}, $\forall i=1,2, \Tilde{f}_{\mathbf{m}^i,{\boldsymbol{\alpha}}^i}$  can be approximated by a fully connected neural network such that its depth is $\mathcal{O}(\ln(1/\delta)+1)$ and width is $\mathcal{O}(\frac{d}{2})$.

Since $(\mathbf{m}^1,{\boldsymbol{\alpha}}^1)$ and $(\mathbf{m}^2,{\boldsymbol{\alpha}}^2)$ both have at most $d^\beta (2N+1)^{\frac{d}{2}}$ possible values,  with $2*d^\beta (2N+1)^{\frac{d}{2}}$ output units, depth $\mathcal{O}(\ln(1/\delta)+1)$ and width $\mathcal{O}( (2N+1)^{\frac{d}{2}})$, by Lemma \ref{le:neural}, we can use a (fully connected) neural network to represent all possible values for $\Tilde{f}_{\mathbf{m}^i,{\boldsymbol{\alpha}}^i}, i=1,2$.

For any pair of $(\mathbf{m}^2,{\boldsymbol{\alpha}}^2)$, if ${\boldsymbol{\alpha}}^1+\boldsymbol{\alpha}^2>\beta$, we set $a_{\mathbf{m},\boldsymbol{\alpha}}=0$.  By equation (\ref{eq:Pm}), $\abs{a_{\mathbf{m},\boldsymbol{\alpha}}} \leq 1$, which satisfies the definition of the function $\Tilde{\times}(\cdot)$.\iffalse Then,
\begin{equation*}
    \sum_{{\mathbf{m}^1,\boldsymbol{\alpha}^1}}a_{\mathbf{m},\boldsymbol{\alpha}}\Tilde{f}_{\mathbf{m}^1,{\boldsymbol{\alpha}}^1}(\mathbf{x}^1) = \sum_{{\mathbf{m}^1,\boldsymbol{\alpha}^1}}\sigma_R(a_{\mathbf{m},\boldsymbol{\alpha}}\Tilde{f}_{\mathbf{m}^1,{\boldsymbol{\alpha}}^1}(\mathbf{x}^1))-\sigma_R(-a_{\mathbf{m},\boldsymbol{\alpha}}\Tilde{f}_{\mathbf{m}^1,{\boldsymbol{\alpha}}^1}(\mathbf{x}^1)).
\end{equation*}

This shows that we can represent the weighted summation with one more layer and maintain the width as $\mathcal{O}((2N+1)^{\frac{d}{2}})$. \fi

Therefore, the full approximation function is redefined as
\begin{equation*}
\begin{split}
      \Tilde{f}&=2^{\frac{d}{2}}d^\beta\sum_{\mathbf{m}^2,\boldsymbol{\alpha}^2} \Tilde{\times}(\frac{1}{2^{\frac{d}{2}}d^\beta}\sum_{{\mathbf{m}^1,\boldsymbol{\alpha}^1}}a_{\mathbf{m},\boldsymbol{\alpha}}\Tilde{f}_{\mathbf{m}^1,\boldsymbol{\alpha}^1}(\mathbf{x}^1),\Tilde{f}_{\mathbf{m}^2,\boldsymbol{\alpha}^2}(\mathbf{x}^2)),
\end{split}
  \end{equation*}
where $\frac{1}{2^{\frac{d}{2}}d^\beta}$ is to ensure the inputs of $\Tilde{\times}(\cdot)$ satisfies the definition.  

By equation (\ref{eq:approximation_error}), for $i=1,2$,
\begin{equation*}
    \abs{\Tilde{f}_{\mathbf{m}^i,\boldsymbol{\alpha}^i}(\mathbf{x}^i)-\phi_{\mathbf{m}^i}(\mathbf{x}^i)(\mathbf{x}^i-\frac{\mathbf{m}^i}{N})^{\boldsymbol{\alpha}^i}}\leq (\frac{d}{2}+\abs{\boldsymbol{\alpha}^i}-1)\delta.
\end{equation*}
Thus,
\begin{equation*}
    \abs{\Tilde{f}(\mathbf{x})-f_1(\mathbf{x})}\leq 2^d d^n(d+\beta)\delta,
\end{equation*}
we choose $\delta = \frac{\epsilon}{2^{d+1}d^{\beta}(d+\beta-1)}$, and $\norm{\Tilde{f}-f}\leq \epsilon$.

By using Lemma \ref{cor:full}, the full approximation neural network is fully connected and has depth $\mathcal{O}(\ln(1/\delta)+1)$, and width at most $\mathcal{O}((2N+1)^{\frac{d}{2}})$. The remainder of the proof follows unchanged. Thus, the weight count scales as $\mathcal{O}((2N+1)^{d}\ln(1/\delta))$.

\textbf{Case 2: if $d$ is odd and $d>1$}, since $d-1$ is even, only the extra dimension requires special treatment.\footnote{We implicitly assume \(N \geq 4\), which is justified since \(N\) becomes large as \(\epsilon\) approaches zero. Therefore, this assumption is made without loss of generality.}

Let $\bar{d}=\frac{d+1}{2}$ denote the central dimension index. Denote $\mathbf{m}^1=(m_1,...,m_{\bar{d}-1})$,$\mathbf{m}^2 = (m_{\bar{d}+1},...,m_d)$, $\boldsymbol{\alpha}^1=(\alpha_1,...,\alpha_{\bar{d}-1})$, $\boldsymbol{\alpha}^2=(\alpha_{\bar{d}+1},...,\alpha_d)$, $\mathbf{x}^1=(x_1,...,x_{\bar{d}-1})$, $\mathbf{x}^2=(x_{\bar{d}+1},...,x_d)$. Similar to the even case, denote $ \Tilde{f}_{\mathbf{m}^1,\boldsymbol{\alpha}^1}(\mathbf{x}^1)$ as the neural network to approximate $f_{\mathbf{m}^1, \boldsymbol{\alpha}^1}(\mathbf{x}^1)\coloneqq\prod_{k=1}^{\bar{d}-1} \psi(3N(x_k-\frac{m_k}{N}))(x_k-\frac{m_k}{N})^{\alpha_k}$ and let $\Tilde{f}_{\mathbf{m}^2,\boldsymbol{\alpha}^2}(\mathbf{x}^2)$ be the neural network to approximate $f_{\mathbf{m}^2, \boldsymbol{\alpha}^2}(\mathbf{x}^2)\coloneqq \prod_{k=\bar{d}+1}^{d} \psi(3N(x_k-\frac{m_k}{N}))(x_k-\frac{m_k}{N})^{\alpha_k}$.

Decompose the function:
\begin{equation*}
\begin{split}
     &\phi_{\mathbf{m}}(\mathbf{x})(\mathbf{x}-\frac{\mathbf{m}}{N})^{\boldsymbol{\alpha}}\\
     =& \prod_{k=1}^d \psi(3N(x_k-\frac{m_k}{N}))(x_k-\frac{m_k}{N})^{\alpha_k}\\
     =& \psi(3N(x_{\bar{d}}-\frac{m_{\bar{d}}}{N}))(x_{\bar{d}}-\frac{m_{\bar{d}}}{N})^{\alpha_{\bar{d}}}f_{\mathbf{m}^1,\boldsymbol{\alpha}^1}(\mathbf{x}^1)f_{\mathbf{m}^2,\boldsymbol{\alpha}^2}(\mathbf{x}^2),
\end{split}
\end{equation*}
since $m_{\bar{d}}\in \{-N,...,-1,0,1,...,N\}$, the key is to separate all $2N+1$ possible values of  $\psi(3N(x_{\bar{d}}-\frac{m_{\bar{d}}}{N}))$ into two parts, each with $\sqrt{2N+1}$ units and the combination of two parts can represent $\psi(3N(x_{\bar{d}}-\frac{m_{\bar{d}}}{N}))$.

Denote $A=\ceil{\sqrt{2N+1}}$. And define sets:
\begin{equation*}
    \begin{split}
        \mathcal{G}_g &= \{-N+i*A+g, i=0,...,A-1\}, g=0,...,A-1.
    \end{split}
\end{equation*}

Each set has $A$ elements. $\forall g=0,...,A-1$, $\forall x_{\bar{d}}\in [-1,1]$, there are at most two $m_{\bar{d}}$, each in separate sets, such that $\psi(3N(x_{\bar{d}}-\frac{m_{\bar{d}}}{N}))>0$.

Also, define
\begin{equation*}
\begin{split}
\mathcal{G}^1&\coloneqq \mathcal{G}_{0}\cup\mathcal{G}_{1}\cup \mathcal{G}_{A-1},\\
   \mathcal{G}^2&\coloneqq \left(\bigcup_{g=2}^{A-2} \mathcal{G}_g \right).
\end{split} 
\end{equation*}

By the construction,
\begin{equation*}
    \begin{split}
        \mathcal{G}^1\cup \mathcal{G}^2 &\supseteq\{-N,...,-1,0,1,...,N\}, \\
         \mathcal{G}_{g'} &\cap \mathcal{G}_g = \emptyset, \quad \forall g,g'=0,...,A-1, g\ne g'.
    \end{split}
\end{equation*}

Define a new function:
\begin{equation*}
    \kappa_b(x_{\bar{d}})=\sum_{j=-N+1+Ab}^{-N+Ab+A-1}\psi(3N(x_{\bar{d}}-\frac{j}{N})), \quad b=0,1,...,A-1.
\end{equation*}

Notice that $j$ doesn't include values from set $\mathcal{G}_0$.

By the construction, $\forall m_{\bar{d}}\in \{-N,...,-1,0,1,...,N\}$, either $ m_{\bar{d}}\in \mathcal{G}^1$, or $ m_{\bar{d}}\in \mathcal{G}^2$. If $ m_{\bar{d}}\in \mathcal{G}^2$, there is a unique $m_{(b,g)}$, such that $m_{\bar{d}}\in [-N+1+Ab,-N+Ab+A-1]$ and $m_{\bar{d}}\in \mathcal{G}_g, g=2,...,A-2$ and $m_{\bar{d}}=m_{(b,g)}$. Notice that there are $(A-3)A$ possible values for $(g,b)$.

Fix a $x_{\bar{d}}\in[-1,1]$, $\exists m \in \{-N,...,-1,0,1,...,N\}$, such that $x_{\bar{d}} \in [\frac{m}{N},\frac{m+1}{N}]$. Then, $\forall m' \in \{-N,...,-1,0,1,...,N\}$, such that $m'\ne m, m'\ne m+1$,
\begin{equation*}
    \psi(3N(x_{\bar{d}}-\frac{m'}{N}))=0.
\end{equation*}

If $m\in \mathcal{G}_g$,
\begin{equation*}
    \psi(3N(x_{\bar{d}}-\frac{m}{N})) = \sum_{i\in \mathcal{G}_g}\psi(3N(x_{\bar{d}}-\frac{i}{N})).
\end{equation*}

If $m+1\in \mathcal{G}_g$,
\begin{equation*}
    \psi(3N(x_{\bar{d}}-\frac{m+1}{N})) = \sum_{i\in \mathcal{G}_g}\psi(3N(x_{\bar{d}}-\frac{i}{N})).
\end{equation*}

If $m,m+1\notin \mathcal{G}_g$,
\begin{equation*}
     \sum_{i\in \mathcal{G}_g}\psi(3N(x_{\bar{d}}-\frac{i}{N}))=0.
\end{equation*}

There are four possible situations. 

First, if $m,m+1\in \mathcal{G}^2$, then $\exists b \in \{0,1,...,A-1\}$,  such that $m,m+1 \in \{-N+1+Ab,...,N+Ab+A-1\}$ and
\begin{equation*}
\begin{split}
    \kappa_b(x_{\bar{d}})&=1,\\
     \kappa_{b'}(x_{\bar{d}})&=0, \quad b'\in  \{0,1,...,A-1\}, b'\ne b.
\end{split}
\end{equation*}

Then,  $\forall b \in \{0,1,...,A-1\}$,
\begin{equation*}
    \begin{split}
        \psi(3N(x_{\bar{d}}-\frac{m_{(b,g)}}{N}))&=\kappa_b(x_{\bar{d}})\sum_{i\in\mathcal{G}_g}\psi(3N(x_{\bar{d}}-\frac{i}{N})), \quad  g =2,...,A-2,\\
        \psi(3N(x_{\bar{d}}-\frac{m_{\bar{d}}}{N}))&=0, \quad \forall m_{\bar{d}}\in \mathcal{G}^1,\\
        \psi(3N(x_{\bar{d}}-\frac{m_{(b,g)}}{N}))&=0, \quad \forall m_{(b,g)}\in\mathcal{G}^2, m_{(b,g)}\ne m,m+1.
    \end{split}
\end{equation*}

For the second situation, if $m,m+1 \in \mathcal{G}^1$, then, $\sum_{i\in\mathcal{G}_g}\psi(3N(x_{\bar{d}}-\frac{i}{N}))=0, g =2,...,A-2$. $\forall b \in \{0,1,...,A-1\}$,
\begin{equation*}
    \begin{split}
        \psi(3N(x_{\bar{d}}-\frac{m_{(b,g)}}{N}))&=\kappa_b(x_{\bar{d}})\sum_{i\in\mathcal{G}_g}\psi(3N(x_{\bar{d}}-\frac{i}{N}))=0, \quad  g =2,...,A-2,\\
        \psi(3N(x_{\bar{d}}-\frac{m_{\bar{d}}}{N}))&= \psi(3N(x_{\bar{d}}-\frac{m_{\bar{d}}}{N})), \quad \forall m_{\bar{d}}\in \mathcal{G}^1.
    \end{split}
\end{equation*}

For the third situation, if $m\in \mathcal{G}^2, m+1\in \mathcal{G}^1$, then $\exists b \in \{0,1,...,A-1\}$,  such that $m,m+1 \in \{-N+1+Ab,...,-N+Ab+A-1\}$ and
\begin{equation*}
\begin{split}
    \kappa_b(x_{\bar{d}})&=1,\\
     \kappa_{b'}(x_{\bar{d}})&=0, \quad b'\in  \{0,1,...,A-1\}, b'\ne b.
\end{split}
\end{equation*}

Then,  $\forall b \in \{0,1,...,A-1\}$,
\begin{equation*}
    \begin{split}
       \psi(3N(x_{\bar{d}}-\frac{m_{(b,g)}}{N}))&=\kappa_b(x_{\bar{d}})\sum_{i\in\mathcal{G}_g}\psi(3N(x_{\bar{d}}-\frac{i}{N})), \quad  g =2,...,A-2,\\
        \psi(3N(x_{\bar{d}}-\frac{m_{\bar{d}}}{N}))&=\psi(3N(x_{\bar{d}}-\frac{m_{\bar{d}}}{N})), \quad \forall m_{\bar{d}}\in \mathcal{G}^1,\\
        \psi(3N(x_{\bar{d}}-\frac{m_{(b,g)}}{N}))&=0, \quad \forall m_{(b,g)}\in\mathcal{G}^2, m_{(b,g)}\ne m,\\
        \psi(3N(x_{\bar{d}}-\frac{m_{\bar{d}}}{N}))&=0, \quad \forall m_{\bar{d}}\in \mathcal{G}^1, m_{\bar{d}}\ne m+1.\\
    \end{split}
\end{equation*}

For the last situation, if $m\in \mathcal{G}^1, m+1\in \mathcal{G}^2$, $\forall b \in \{0,1,...,A-1\}$,
\begin{equation*}
    \begin{split}
       \psi(3N(x_{\bar{d}}-\frac{m_{(b,g)}}{N}))&=\kappa_b(x_{\bar{d}})\sum_{i\in\mathcal{G}_g}\psi(3N(x_{\bar{d}}-\frac{i}{N})), \quad  g =2,...,A-2,\\
        \psi(3N(x_{\bar{d}}-\frac{m_{\bar{d}}}{N}))&=\psi(3N(x_{\bar{d}}-\frac{m_{\bar{d}}}{N})), \quad \forall m_{\bar{d}}\in \mathcal{G}^1,\\
        \psi(3N(x_{\bar{d}}-\frac{m_{(b,g)}}{N}))&=0, \quad \forall m_{(b,g)}\in\mathcal{G}^2, m_{(b,g)}\ne m+1,\\
        \psi(3N(x_{\bar{d}}-\frac{m_{\bar{d}}}{N}))&=0, \quad \forall m_{\bar{d}}\in \mathcal{G}^1, m_{\bar{d}}\ne m.\\
    \end{split}
\end{equation*}

In conclusion, if  $m_{\bar{d}}\in \mathcal{G}^1$,
\begin{equation*}
\begin{split}
     &\phi_{\mathbf{m}}(\mathbf{x})(\mathbf{x}-\frac{\mathbf{m}}{N})^{\boldsymbol{\alpha}}\\
     &=\psi(3N(x_{\bar{d}}-\frac{m_{\bar{d}}}{N}))(x_{\bar{d}}-\frac{m_{\bar{d}}}{N})^{\alpha_{\bar{d}}}\Tilde{f}_{\mathbf{m}^1,\boldsymbol{\alpha}^1}(\mathbf{x}^1)\Tilde{f}_{\mathbf{m}^2,\boldsymbol{\alpha}^2}(\mathbf{x}^2), 
\end{split}
\end{equation*}
if $m_{\bar{d}}\in \mathcal{G}^2, m_{\bar{d}}= m_{(b,g)}$,
\begin{equation*}
\begin{split}
     &\phi_{\mathbf{m}}(\mathbf{x})(\mathbf{x}-\frac{\mathbf{m}}{N})^{\boldsymbol{\alpha}}\\
     &= \psi(3N(x_{\bar{d}}-\frac{m_{\bar{d}}}{N}))(x_{\bar{d}}-\frac{m_{\bar{d}}}{N})^{\alpha_{\bar{d}}}\Tilde{f}_{\mathbf{m}^1,\boldsymbol{\alpha}^1}(\mathbf{x}^1)\Tilde{f}_{\mathbf{m}^2,\boldsymbol{\alpha}^2}(\mathbf{x}^2)\\
     &=\left(\sum_{i\in\mathcal{G}_g}\psi(3N(x_{\bar{d}}-\frac{i}{N}))(x_{\bar{d}}-\frac{i}{N})^{\alpha_{\bar{d}}}\Tilde{f}_{\mathbf{m}^1,\boldsymbol{\alpha}^1}(\mathbf{x}^1)\right)\left(\Tilde{f}_{\mathbf{m}^2,\boldsymbol{\alpha}^2}(\mathbf{x}^2)\kappa_b(x_{\bar{d}})\right),
\end{split}
\end{equation*}

Adjusting the definition,
\begin{equation*}
    \begin{split}\Tilde{f}_{\mathbf{m}^1,\boldsymbol{\alpha}^1,g,\alpha_{\bar{d}}}(\mathbf{x}^1,x_{\bar{d}})&\coloneqq\Tilde{\times}\left(\sum_{i\in\mathcal{G}_g}\Tilde{\times}(\psi(3N(x_{\bar{d}}-\frac{i}{N})),\Tilde{\times}(x_{\bar{d}}-\frac{i}{N},\dots )),\Tilde{f}_{\mathbf{m}^1,\boldsymbol{\alpha}^1}(\mathbf{x}^1)\right),\\
\Tilde{f}_{\mathbf{m}^2,\boldsymbol{\alpha}^2,b}(\mathbf{x}^2,x_{\bar{d}})&\coloneqq\Tilde{\times}\left(\Tilde{f}_{\mathbf{m}^2,\boldsymbol{\alpha}^2}(\mathbf{x}^2),\kappa_b(x_{\bar{d}})\right).
    \end{split}
\end{equation*}

For each neural network, with at most $(d-1)^\beta (2N+1)^{\frac{d-1}{2}}\ceil{\sqrt{2N+1}}$ output units, depth $\mathcal{O}(\ln(1/\delta)+1)$ and width $\mathcal{O}((2N+1)^{\frac{d}{2}})$, by Lemma \ref{le:neural}, we can use a neural network to represent all possible values.

For each $(\mathbf{m}^2,\boldsymbol{\alpha}^2)$ and $(\mathbf{m}^1,\boldsymbol{\alpha}^1)$, if $m_{\bar{d}}\in \mathcal{G}^2$, let $a_{\mathbf{m,\boldsymbol{\alpha}},(b,g),\alpha_{\bar{d}}}$ denote $a_{\mathbf{m},\boldsymbol{\alpha}}$ in the case $m_{\bar{d}}=m_{(b,g)}$; if $m_{\bar{d}}\in \mathcal{G}^1$, let $a_{\mathbf{m,\boldsymbol{\alpha}}}$ denote $a_{\mathbf{m},\boldsymbol{\alpha}}$.  We set $a_{\mathbf{m,\boldsymbol{\alpha}},(b,g),\alpha_{\bar{d}}}=0$ if it doesn't satisfy the definition. We can rewrite $f_1(\mathbf{x})$ as
\begin{equation*}\label{eq:full}
\begin{split}
   f_1(\mathbf{x}) = 
    \sum_{\mathbf{m}^2, \boldsymbol{\alpha}^2, b} 
    \Bigg( &
            \sum_{\mathbf{m}^1, \boldsymbol{\alpha}^1,g \in \mathcal{G}^2,\alpha_{\bar{d}}} 
            a_{\mathbf{m}^1, \boldsymbol{\alpha}^1,(b,g),\alpha_{\bar{d}}}
            f_{\mathbf{m}^1, \boldsymbol{\alpha}^1, g,\alpha_{\bar{d}}}(\mathbf{x}^1, x_{\bar{d}}) \Bigg)\\
       &\times f_{\mathbf{m}^2, \boldsymbol{\alpha}^2, b}(\mathbf{x}^2, x_{\bar{d}})
     \\
    \quad \displaystyle+
    \sum_{\mathbf{m}^2, \boldsymbol{\alpha}^2} 
            \Bigg(&
            \sum_{{\mathbf{m}^1, \boldsymbol{\alpha}^1, m_{\bar{d}} \in \mathcal{G}^1},\alpha_{\bar{d}}} 
            a_{\mathbf{m}, \boldsymbol{\alpha}} 
            \left(x_{\bar{d}} - \frac{m_{\bar{d}}}{N} \right)^{\alpha_{\bar{d}}}
            \psi\left( 3N\left(x_{\bar{d}} - \frac{m_{\bar{d}}}{N} \right) \right)
            f_{\mathbf{m}^1, \boldsymbol{\alpha}^1}(\mathbf{x}^1)\Bigg)\\
        &\times f_{\mathbf{m}^2, \boldsymbol{\alpha}^2}(\mathbf{x}^2).
\end{split}
\end{equation*}

Define the full approximation by the following:
\begin{equation*}\label{eq:full}
\begin{split}
   \Tilde{f} = 2^{\frac{d}{2}} d^\beta 
    \sum_{\mathbf{m}^2, \boldsymbol{\alpha}^2, b} 
    \Tilde{\times}\Bigg( &
            \frac{1}{2^{\frac{d}{2}} d^\beta}
            \sum_{\mathbf{m}^1, \boldsymbol{\alpha}^1,g \in \mathcal{G}^2,\alpha_{\bar{d}}} 
            a_{\mathbf{m}^1, \boldsymbol{\alpha}^1,(b,g),\alpha_{\bar{d}}}
            \Tilde{f}_{\mathbf{m}^1, \boldsymbol{\alpha}^1, g,\alpha_{\bar{d}}}(\mathbf{x}^1, x_{\bar{d}}), \\
       & \Tilde{f}_{\mathbf{m}^2, \boldsymbol{\alpha}^2, b}(\mathbf{x}^2, x_{\bar{d}})
    \Bigg) \\
    \quad \displaystyle+ 2^{\frac{d}{2}} d^\beta 
    \sum_{\mathbf{m}^2, \boldsymbol{\alpha}^2} 
    \Tilde{\times}\Bigg( 
        &\Tilde{f}_{\mathbf{m}^2, \boldsymbol{\alpha}^2}(\mathbf{x}^2),\\
         \frac{1}{2^{\frac{d}{2}} d^\beta}
            \sum_{{\mathbf{m}^1, \boldsymbol{\alpha}^1, m_{\bar{d}} \in \mathcal{G}^1},\alpha_{\bar{d}}} &
            a_{\mathbf{m}, \boldsymbol{\alpha}} 
            \Tilde{\times}\bigg(\Tilde{f}_{\mathbf{m}^1, \boldsymbol{\alpha}^1}(\mathbf{x}^1),\Tilde{\times}\left(\psi(3N(x_{\bar{d}}-\frac{m_{\bar{d}}}{N})),\Tilde{\times}(x_{\bar{d}}-\frac{m_{\bar{d}}}{N},\dots )\right)\bigg)
    \Bigg).
\end{split}
\end{equation*}

For the neural network in the first and second lines, the depth is $\mathcal{O}(\ln(1/\delta)+1)$ and width is $\mathcal{O}((2N+1)^{\frac{d}{2}})$; since $\mathcal{G}^1$ has at most $3 \ceil{\sqrt{2N+1}}$ elements, we need neural networks in the third and fourth lines to have depth $\mathcal{O}(\ln(1/\delta)+1)$ and width $\mathcal{O}( (2N+1)^{\frac{d}{2}})$. To construct the first line's neural network, we need one more layer of composition $\Tilde{\times}(\cdot)$. Thus, we must adjust the term $d+\beta-1$ in Appendix \ref{proof:error} into $d+\beta$. That is, we choose $\delta = \frac{\epsilon}{2^{d+1}d^\beta (d+\beta)\delta}$. The rest of the proof remains the same as in Appendix \ref{proof:error}.

\textbf{Case 3: if $d=1$}, we need a special treatment.\footnote{Although there is no theoretical advantage in using deep neural networks when $d=1$, we include a brief proof for completeness.} The method used above requires dividing $[-1,1]$ into $N$ grids and $N$ different functions of $\psi(x)$, the width is $\mathcal{O}(N)$. Since $d\geq 2$, the largest width is $\mathcal{O}(N^{\frac{d}{2}})$, this approach is applicable to the previous cases. However, when $d=1$, we need the largest width to be $\mathcal{O}(\sqrt{N})$. We will follow the notation used above.

First, notice that the function $\kappa_b(x)$ requires a width $\mathcal{O}(N)$; we first modify this function.

Define a new function, $\forall b=0,1,\dots,A-1$,
\begin{equation*}
        \Psi_b(x)=\begin{cases}
            1, & x\in [-N+Ab+1,-N+Ab+A-1],\\
            \psi(3N(x-\frac{-N+Ab+1}{N})), & x\in [-N+Ab,-N+Ab+1],\\
            \psi(3N(x-\frac{-N+Ab+A-1}{N})), & x\in [-N+Ab+A-1,-N+Ab+A],\\
            0, &\text{otherwise}.
        \end{cases}
    \end{equation*}
This function can be directly represented by a ReLU neural network with finite width and depth. Notice that

\begin{equation*}
    \Psi_b(x)=\kappa_b(x).
\end{equation*}
To represent all possible values of $\Psi_b(x_{\bar{d}})$, we only need a neural network with finite depth and width $\mathcal{O}(\sqrt{N})$.

$\forall m'\in \mathcal{G}^2, \exists m \in [-N+2,-N+A-2], \exists b=0,1,\dots,A-1$, such that $m' = m+b$. We denote $a_{m',\alpha}$ as $a_{m,b,\alpha}$. We notice that $f_1(x)$ can be written as:
\begin{equation*}
\begin{split}
    f_1(x) &= \sum_b\sum_\alpha \sum_{m=-N+2}^{-N+A-2}  a_{m,b,\alpha} \Psi_b(x)\cdot \psi( x-\frac{m+bA}{N})( x-\frac{m+bA}{N})^{\alpha}\\
    &+\sum_\alpha \sum_{m\in \mathcal{G}^1} a_{m,\alpha} \psi( x-\frac{m}{N})(x-\frac{m}{N})^{\alpha}\\
    &= \sum_b  \Psi_b(x)\cdot  \sum_{m=-N+2}^{-N+A-2}\sum_\alpha  a_{m,b,\alpha}\psi\left(-\frac{m}{N}+\sum_{b'} \Psi_{b'}(x)\cdot (x-\frac{{b'}A}{N})\right)\\
    &\cdot( -\frac{m}{N} +\sum_{b'} \Psi_{b'}(x)\cdot (x-\frac{{b'}A}{N}))^{\alpha}\\
    &+\sum_\alpha \sum_{m\in \mathcal{G}^1} a_{m,\alpha} \psi( x-\frac{m}{N})(x-\frac{m}{N})^{\alpha}
\end{split}
\end{equation*}

Define the full approximation by the following:
\begin{equation*}
    \begin{split}
    \Tilde{f} &=  2d^{\beta} \sum_b \Tilde{\times}\Biggl( \Psi_b(x),\\
    &\frac{1}{2d^{\beta}}\sum_{m=-N+2}^{-N+A-2}\sum_\alpha  a_{m,b,\alpha} \Tilde{\times}\biggl(\psi\Bigl(-\frac{m}{N}+ \sum_{b'} \Tilde{\times}\bigl( \Psi_{b'}(x),x-\frac{{b'}A}{N}\bigr)\Bigr),\\
    &\Tilde{\times}\Bigl( -\frac{m}{N}+ \sum_{b'} \Tilde{\times}\bigl( \Psi_{b'}(x),x-\frac{{b'}A}{N}\bigr),\dots\Bigr)\biggr)\Biggr)\\
    &+\sum_\alpha \sum_{m\in \mathcal{G}^1} a_{m,\alpha}  \Tilde{\times}\left(\Psi_b(x), \Tilde{\times}\Bigl(\psi( x-\frac{m}{N}),\Tilde{\times}( x-\frac{m}{N},\dots)\Bigr)\right)
\end{split}
\end{equation*}
The width needed for this approximation is at most $\mathcal{O}(\sqrt{N})$ and depth $\mathcal{O}(\ln(1/\delta)+1)$. The rest of the proof remains the same as above.

In conclusion, there is a ReLU network that has depth $c(\ln(1/\epsilon)+1)$, width at most $c\epsilon^{-\frac{d}{2\beta}}$, and weights at most $c\epsilon^{-\frac{d}{\beta}}(\ln(1/\epsilon)+1)$, with some constant $c=c(d,\beta)$. Lemma \ref{cor:full} shows that this network can be fully connected with the same level of depth, width, and weights. The only difference is the constant term $c$.

\end{proof} \end{appendix}

%%%%%%%%%%%%%%%%%%%%%%%%%%%%%%%%%%%%%%%%%%%%%%
%% Bibliography:                            %%
%%%%%%%%%%%%%%%%%%%%%%%%%%%%%%%%%%%%%%%%%%%%%%
%% IMPORTANT: References in the bibliography should be complete, 
%% including the first and last names, and date of publication.
\bibliographystyle{apalike}
\bibliography{bib}  % Bibliography file (usually '*.bib')

@article{poggio2017and,
  title={Why and when can deep-but not shallow-networks avoid the curse of dimensionality: a review},
  author={Poggio, Tomaso and Mhaskar, Hrushikesh and Rosasco, Lorenzo and Miranda, Brando and Liao, Qianli},
  journal={International Journal of Automation and Computing},
  volume={14},
  number={5},
  pages={503--519},
  year={2017},
  publisher={Springer}
}

@article{shaham2018provable,
  title={Provable approximation properties for deep neural networks},
  author={Shaham, Uri and Cloninger, Alexander and Coifman, Ronald R},
  journal={Applied and Computational Harmonic Analysis},
  volume={44},
  number={3},
  pages={537--557},
  year={2018},
  publisher={Elsevier}
}

@article{stone1982optimal,
  title={Optimal global rates of convergence for nonparametric regression},
  author={Stone, Charles J},
  journal={The annals of statistics},
  pages={1040--1053},
  year={1982},
  publisher={JSTOR}
}

@article{farrell2021deep,
  title={Deep neural networks for estimation and inference},
  author={Farrell, Max H and Liang, Tengyuan and Misra, Sanjog},
  journal={Econometrica},
  volume={89},
  number={1},
  pages={181--213},
  year={2021},
  publisher={Wiley Online Library}
}

@article{shen2021deep,
  title={Deep quantile regression: Mitigating the curse of dimensionality through composition},
  author={Shen, Guohao and Jiao, Yuling and Lin, Yuanyuan and Horowitz, Joel L and Huang, Jian},
  journal={arXiv preprint arXiv:2107.04907},
  year={2021}
}

@article{yarotsky2017error,
  title={Error bounds for approximations with deep ReLU networks},
  author={Yarotsky, Dmitry},
  journal={Neural networks},
  volume={94},
  pages={103--114},
  year={2017},
  publisher={Elsevier}
}

@article{liu2022optimal,
  title={Optimal nonparametric inference via deep neural network},
  author={Liu, Ruiqi and Boukai, Ben and Shang, Zuofeng},
  journal={Journal of Mathematical Analysis and Applications},
  volume={505},
  number={2},
  pages={125561},
  year={2022},
  publisher={Elsevier}
}

@article{schmidt2020nonparametric,
  title={Nonparametric regression using deep neural networks with ReLU activation function},
  author={Schmidt-Hieber, Johannes},
journal={The Annals of Statistics},
  year={2020}
}

@article{petersen2024mathematical,
  title={Mathematical theory of deep learning},
  author={Petersen, Philipp and Zech, Jakob},
  journal={arXiv preprint arXiv:2407.18384},
  year={2024}
}

@article{bauer2019deep,
  title={On deep learning as a remedy for the curse of dimensionality in nonparametric regression},
  author={Bauer, Benedikt and Kohler, Michael},
journal={The Annals of Statistics},
  year={2019}
}

@book{milnor1997topology,
  title={Topology from the differentiable viewpoint},
  author={Milnor, John Willard and Weaver, David W},
  volume={21},
  year={1997},
  publisher={Princeton university press}
}

@article{kohler2022estimation,
  title={Estimation of a function of low local dimensionality by deep neural networks},
  author={Kohler, Michael and Krzy{\.z}ak, Adam and Langer, Sophie},
  journal={IEEE transactions on information theory},
  volume={68},
  number={6},
  pages={4032--4042},
  year={2022},
  publisher={IEEE}
}

@article{stone1994use,
  title={The use of polynomial splines and their tensor products in multivariate function estimation},
  author={Stone, Charles J},
  journal={The annals of statistics},
  pages={118--171},
  year={1994},
  publisher={JSTOR}
}

@article{feng2023over,
  title={Over-parameterized deep nonparametric regression for dependent data with its applications to reinforcement learning},
  author={Feng, Xingdong and Jiao, Yuling and Kang, Lican and Zhang, Baqun and Zhou, Fan},
  journal={Journal of Machine Learning Research},
  volume={24},
  number={383},
  pages={1--40},
  year={2023}
}

@article{brown2024statistical,
  title={Statistical Properties of Deep Neural Networks with Dependent Data},
  author={Brown, Chad},
  journal={arXiv preprint arXiv:2410.11113},
  year={2024}
}

@article{zhang2024causal,
  title={Causal inference through multi-stage learning and doubly robust deep neural networks},
  author={Zhang, Yuqian and Bradic, Jelena},
  journal={arXiv preprint arXiv:2407.08560},
  year={2024}
}

@article{colangelo2025double,
  title={Double debiased machine learning nonparametric inference with continuous treatments},
  author={Colangelo, Kyle and Lee, Ying-Ying},
  journal={Journal of Business \& Economic Statistics},
  number={just-accepted},
  pages={1--26},
  year={2025},
  publisher={Taylor \& Francis}
}

@article{zhang2024classification,
  title={Classification with deep neural networks and logistic loss},
  author={Zhang, Zihan and Shi, Lei and Zhou, Ding-Xuan},
  journal={Journal of Machine Learning Research},
  volume={25},
  number={125},
  pages={1--117},
  year={2024}
}

@article{chronopoulos2023deep,
  title={Deep neural network estimation in panel data models},
  author={Chronopoulos, Ilias and Chrysikou, Katerina and Kapetanios, George and Mitchell, James and Raftapostolos, Aristeidis},
  journal={arXiv preprint arXiv:2305.19921},
  year={2023}
}

@article{jiao2025deep,
  title={Deep approximate policy iteration},
  author={Jiao, Yuling and Kang, Lican and Liu, Jin and Lu, Xiliang and Yang, Jerry Zhijian},
  journal={The Annals of Statistics},
  volume={53},
  number={2},
  pages={802--821},
  year={2025},
  publisher={Institute of Mathematical Statistics}
}

\end{document}